\newif\ifFull
\Fullfalse
\documentclass{cccg14}
\usepackage{amsmath, amssymb, algorithm, algorithmic, color}
\usepackage{enumerate, graphicx, hyperref}
\graphicspath{{figures/}}
\usepackage{cite}
\usepackage{wrapfig}
\usepackage{todonotes}

\usepackage{mathptm}
\DeclareSymbolFont{largesymbols}{OMX}{cmex}{m}{n}

\newcommand{\R}{\mathbf{R}}

\renewcommand{\subsection}[1]{\paragraph{\textbf{\normalsize #1.}}}

\title{
    Windows into Geometric Events:\\
    Data Structures for Time-Windowed \\
Querying of Temporal Point Sets}
\author{
    Michael~J.~Bannister
\thanks{
Dept. of Comp. Sci.,
U.~of CA, Irvine, 
\texttt{mbannist(at)uci.edu}
}
    \and
William~E.~Devanny 
\thanks{
Dept. of Comp. Sci.,
U.~of CA, Irvine, 
\texttt{wdevanny(at)uci.edu}
}
\and
Michael~T.~Goodrich 
\thanks{
Dept. of Comp. Sci.,
U.~of CA, Irvine, 
\texttt{goodrich(at)uci.edu}
}
\and
Joseph A. Simons
\thanks{
Dept. of Comp. Sci.,
U.~of CA, Irvine, 
\texttt{jsimons(at)uci.edu}
}
\and
Lowell~Trott
\thanks{
Dept. of Comp. Sci.,
U.~of CA, Irvine, 
\texttt{ltrott(at)uci.edu}
}
}

\index{Bannister, Michael}
\index{Devanny, William E.}
\index{Goodrich, Michael T.}
\index{Simons, Joseph A.}
\index{Trott, Lowell}

\begin{document}
\date{}
\maketitle
\begin{abstract}
\ifFull 
We study geometric data structures for sets of point-based temporal events, such as geo-tagged scientific observations, sporting actions, or disaster responses, so as to answer real-time \emph{time-windowed} queries, each of which is specified by a range of event times, $[t_1,t_2]$, and a geometric question to be answered for the contiguous subsequence of point events with time stamps in this range. In particular, we consider time-windowed queries for skylines in $\mathbf{R}^d$, convex hulls in $\mathbf{R}^2$ and proximity based queries in $\mathbf{R}^d$. For skyline problems, we provide a data structure using near-linear space and preprocessing that can report the skyline in time proportional to its size. For convex hull problems, we provide a data structures using near-linear space and preprocessing that can in polylogarithmic time answer various convex hull queries, including gift wrapping, hull reporting, tangent finding, linear programming, line decision, line stabbing, membership, and containment. For proximity problems, we again provide data structures using near-linear space and preprocessing that can in polylogarithmic time answer approximate spherical range queries, approximate nearest neighbor queries, and in linear time construct many proximity graphs for a given query window.
\else
We study geometric data structures for sets of point-based temporal events, answering \emph{time-windowed} queries, i.e., given a contiguous time interval we answer common geometric queries about the point events with time stamps in this interval. The geometric queries we consider include queries based on the skyline, convex hull, and proximity relations of the point set. We provide space efficient data structures which answer queries in polylogarithmic time.
\fi
\end{abstract}

\section{Introduction}
Spatio-temporal data sets deal with geometric objects associated 
with events occurring at specific times and places
(e.g., see~\cite{Erwig:1998,Roddick:1999}).
\ifFull
Examples include geo-tagged events associated with animal 
movements~\cite{hqs-seda-00}, sports~\cite{abcb-sasv-02},
and disaster responses~\cite{mbkv-pr-03}.
\fi
Thus, we consider an \emph{event} in this context
to be a triple, $(t,p,v)$, where $t$ is a 
time stamp of occurrence for this event, $p$ is a point in $\mathbf{R}^d$ 
describing the location of this event,
and $v$ is a set of 
additional data values that may also be associated with this event.
Exploring such spatio-temporal data sets is facilitated by data structures
that support queries involving these spatial and temporal attributes.

In 
a \emph{time-windowed query}~\cite{BanDubEppSmy-SODA-13}, 
we are given an interval of 
time, $[t_1,t_2]$, and a predicate, $\cal P$, and we are
interested in the events matching $\cal P$ which have time stamps in the
range $[t_1,t_2]$.
\ifFull
In particular, we are interested in the design of efficient
data structures 
for answering time-windowed geometric
queries defined on points associated with the events in a given time window.
\else
\fi
Formally, we preprocess a sequence of points 
$p_k$ in $\mathbf{R}^d$ for $0 \leq k \leq n - 1$ 
in order to answer queries on \emph{windows}
into this sequence of points, where a window is of the form 
$[p_i, p_j] = \{p_k \mid i \leq k \leq j\}$. 
We require that the runtime of our queries depends only
on the \emph{width} of the window $w = j - i + 1$, and not on $n$ the total
number of temporal points.  We assume there is a polynomial-sized set, $U$, of
identifiable moments in time, based on a reasonable way of measuring time.
E.g., each nanosecond from the birth of the sun until its projected death
can be indexed using a 64-bit integer.

Previous data structure frameworks involving point data have taken
various approaches with respect to time and location.
Traditionally, this has been broken down into two 
approaches---a \emph{static} approach, where one assumes that all the input
points are given simultaneously at ``time zero'' and then queries
are performed on this set, and 
a \emph{dynamic} approach, where points are inserted and deleted over time
and queries are performed with respect to the current 
\ifFull 
set
(e.g., see~\cite{ct-dacg-92,de2008computational}). 
\else
set.
\fi
Motivated by geo-tagging applications,
\ifFull
e.g., see~\cite{abcb-sasv-02,hqs-seda-00,mbkv-pr-03}, 
\fi
we take an event-based approach, where
point-based events appear at specific instances in time,
so that, at any point in time, there is at most a single point that exists
in our data set.
Thus, it is only
by considering \ifFull point events in\fi{} time intervals
that we get sets of points over which we can ask geometric queries. 

Of course, if the queries are themselves axis-aligned range queries, then time-windowed
queries can be answered simply by considering time as yet another 
dimension and storing the events in a $(d+1)$-dimensional range-searching 
data structure.
This approach does not carry over, however, to convex hull queries, proximity
queries, or skyline queries.
There are previous data structure approaches
that have nevertheless considered other variations
with respect to time, updates, and queries.
In batched dynamic querying, for instance,
a set of queries is given in advance for a 
static set of points~\cite{Edelsbrunner1985515}, 
and in off-line geometric querying, queries are performed 
in ``the past'' with respect to a pre-specified sequence of 
updates~\cite{Agarwal199165,hs-olmpc-96}. 
In time-windowed querying, on the other hand,
we don't know the time windows or the queries being requested 
in those windows in advance, and we do not restrict 
ourselves to windows starting at ``time zero.''

Likewise, time-windowed querying is 
not the same as the persistent 
data structure framework (e.g., see~\cite{Driscoll198986}),
 where a sequence of insertions and deletions is performed 
on a data structure in an online fashion, with that data structure 
adapted to allow for queries coming later that are done ``in the past.''
Time-windowed data structures allow for different ``starting times'' 
for such sequences of operations, 
whereas these previous persistent approaches start at 
``time zero.''

Our model is also \ifFull related to but fundamentally\fi{} different from previous work
on geometric querying on concatenable structures 
for ordered decomposable problems 
(e.g., see~\cite{Grossi19991,Vankreveld1994130}).
In this framework, a set of objects is given in some order, such as
$x$-coordinates, and geometric queries are performed on this set, subject to splits and merges along one of these dimensions.
However, it's not clear how to apply these previous
approaches to time-windowed queries,
since their results depend on decomposing the data set based on the 
\emph{geometry} of the points, whereas the time-windowed framework
instead supports queries based on the
\emph{time-stamps} of the points, which are unrelated to their geometry. 

Our approach is also related to but distinct from previous work
on \emph{kinetic data structures} (e.g., see~\cite{Basch:1997}).
In this framework, each point \ifFull in a geometric data set\fi{} has a given trajectory
that describes its movement over time, subject to trajectory updates and
queries involving point configurations that would exist
at given times based on the current set of trajectories.
In the time-windowed framework, on the other hand, points
exist only as events that occur at specific times; hence, they
are not given with trajectories.

Chan~\cite{c-dpchnlat-01} applies a query oriented approach to maintaining a dynamic convex
hull. 
However,
his structure 
is only suited
to answering queries on the current set of data points, not on windows in time.  

Perhaps the most closely related prior work 
for geometric data is that of Shi and JaJa~\cite{ShiJaj-ISAAC-04},
 which also considers 
geometric queries on time-windows of temporal data. 
However, they only consider 
``conjunctive temporal range search'' queries, 
which are fundamentally different than
the types of queries we consider.

Time-windowed querying was also considered by
Bannister {\it et al.}~\cite{BanDubEppSmy-SODA-13},
for relational network data and queries based on graph-theoretic primitives.
They give a number of efficient 
data structures for answering such queries, and we adapt some of their methods to the problem of computing skylines.
However, their methods
do not translate into efficient data structures for convex hulls or proximity queries.

\subsection{Our Results}
In this paper, we study data structures for geometric queries based on the
skyline, proximity, and convex hull of points in the time-windowed query model.
If the width of the query window is fixed, data structures supporting windowed queries can be built using existing persistent data structures. The difficulty here, however, is that the window must be fixed in advance, which is not typically useful for data exploration purposes. For this reason we place no \emph{a priori} restrictions on the query windows.

We consider the problem of performing
convex hull queries on points in $\mathbf{R}^2$ within a time window. 
Previously, the problem of reporting  the convex hull of points within 
a two-dimensional query rectangle 
has been considered~\cite{Vankreveld1994130, Grossi19991}, 
but such results do not extend to our time-windowed queries, of course.
The method used here is based instead on hierarchical decomposition in time. 
We build a data
structure of size $O(n \log n)$ in time $O(n \log n)$ from which we can answer 
time-windowed queries based on the 
convex hull of points in the window, in polylogarithmic time\ifFull
(see Table~\ref{tab:convex})\fi.

We also
modify the decomposition tree used for convex hull queries to answer windowed proximity based queries. We develop data structures answering approximate nearest neighbor queries, using near-linear space and polylogarithmic query time. In addition, we develop data structures to construct proximity based graphs for a given window, e.g., Delaunay triangulation, minimum spanning tree, nearest neighbor graph and Gabriel graph, in near-linear space and linear time.

Finally, we consider the problem of reporting the skyline, and for colored points the problem of reporting the set of unique colors on the skyline. 
Computing the skyline of a data set is 
classically known as the \emph{maxima set problem}~\cite{Kung:1975} and is
important in multi-criteria decision 
making~\cite{BorKosSto-ICDE-2001, MCO-Surv}\ifFull,\else.\fi
\ifFull
e.g., when purchasing a car, one may consider a car's 
affordability (inverse of price) and performance. 
If one car dominates another in both affordability and performance,
then the dominated car need not be considered. For this reason we should only consider cars appearing on the skyline. If we color the cars by their make, then a query for the distinct colors on the skyline will tell use which makes have a model that is dominant in both affordability and performance.
\else
\fi
\ifFull 
We build a data structure of size $O(n^{1+\epsilon})$ in time $O(n^{1+\epsilon})$ that
outputs the skyline, distinct colors on the skyline, or size of the skyline for a given time window in time proportional to the output. We achieve these results by adapting the methods used by Bannister {\it et al.}~\cite{BanDubEppSmy-SODA-13}.
\else
We achieve these results by adapting the methods used by Bannister {\it et
al.}~\cite{BanDubEppSmy-SODA-13}. Due to space constraints, details of this and
other theorems and lemmas are given in the appendix.
\begin{theorem}
A sequence of temporal points, $p_i$ for $0 \leq i < n$, in $\mathbf{R}^d$ can be preprocessed into a data structure of size $O(n^{1+\epsilon})$ in $O(n^{1+\epsilon})$ time such that a query for the skyline of $[p_i, p_j]$ can be reported in $O(k)$. Furthermore, if the points are colored then the distinct colors on the skyline can be reported in time $O(k)$.
\end{theorem}
\fi

\ifFull
\begin{table}
\centering
\small
\vspace{-1em}
\begin{tabular}{|lc|}
\hline
Gift wrap \hspace{10em} &  $O(\log^2 w)$\\
Reporting & $O(h\log^2 w)$\\
Tangent query & $O(\log^2w)$\\
Linear programming & $O(\log w)$\\
Line decision & $O(\log w)$\\
Line stabbing & $O(\log^2 w)$\\
Vertical line stabbing & $O(\log^2 w)$\\
Containment & $O(\log^2 w)$\\
Membership & $O(\log^2 w)$ \\
\hline
\end{tabular}
\caption{Running times for convex hull queries, where $w$ and $h$ are the sizes of the window and convex hull.}
\label{tab:convex}
\vspace{-1em}
\end{table}
\fi

\ifFull
\section{Skyline}\label{app:skyline}
\ifFull
\begin{figure}[t]
\vspace{-1em}
\begin{minipage}[t]{0.5\linewidth}
\centering
\includegraphics[height=1.25in]{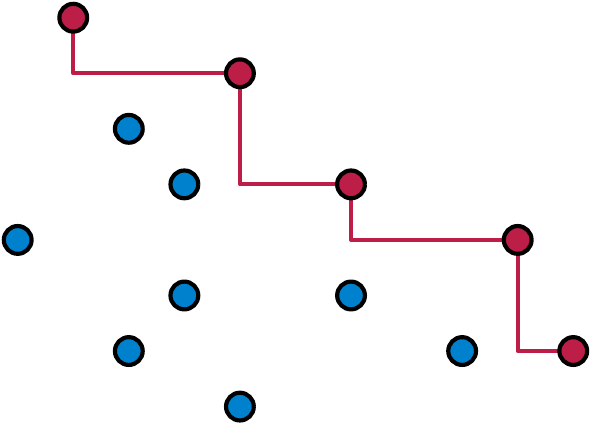}
\vspace{-1em}
\caption{The skyline in this set of points.}
\label{fig:skyline-example}
\end{minipage}
\begin{minipage}[t]{0.5\linewidth}
\centering
\includegraphics[height=1.25in]{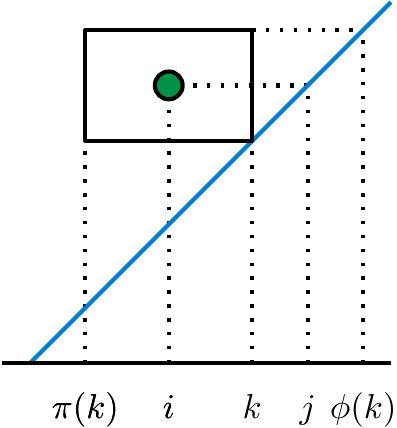}
\vspace{-1em}
\caption{Rectangle stabbing inequality}
\label{fig:rectangle-stabbing}
\end{minipage}
\vspace{-1em}
\end{figure}
\else
\begin{wrapfigure}{r}{1.8in}
\centering
\vspace{-2em}
\includegraphics[width=1.7in]{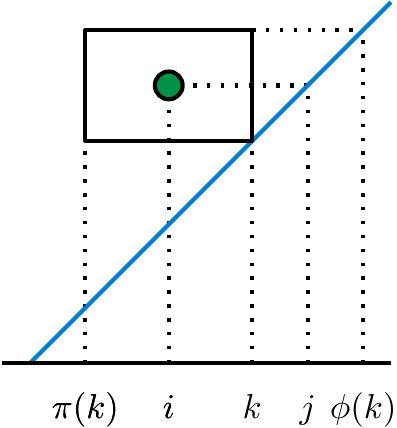}
\vspace{-1em}
\caption{Stabbing inequality}
\vspace{-1em}
\label{fig:rectangle-stabbing}
\end{wrapfigure}
\fi
In this section, we form a general method for answering queries for the set of maximal elements under a preordered relation and then specialize this method to the problem of computing the skyline \ifFull(See Fig.~\ref{fig:skyline-example})\fi. Previously the problem of computing the skyline for a fixed-width sliding window was considered by Lin {\it et al.}~\cite{LinYuaWan-ICDE-2005}. The fastest algorithms for skyline run in time $O(n\log^{d-3} n)$ 
in the worst case for points in $\mathbf{R}^d$ for fixed $d \geq 4$~\cite{ChaLarPat-SOCG-2011}, and the fastest output sensitive algorithms run in time $O(n\log^{d-2} k)$ for fixed $d \geq 3$~\cite{KirSeid-SOCG-1985}.

Given a set, $S,$ a binary relation, $<$, is said to be an \emph{(irreflexive) preorder} if (1) for
all $a \in S$, $a \not < a$ (irreflexive); (2) for all $a,b,c \in S$ if $a < b$ and $b < c$,
then $a < c$ (transitive). An element, $x$, in a subset, $E$, of $S$ is said to be \emph{maximal}
in $E$ if there does not exist a $y \in E$ with $x < y$. Finally, given a
sequence, $e_i$ for $0 \leq i < n$, of elements from a preordered set $S$ we
define for each element $e_k$ the function $\phi(k)$ to be the largest $j'$ such
that $e_k \not< e_j$ whenever $k \leq j \leq j'$; and, $\pi(k)$ to be the
smallest $i'$ such that $e_k \not< e_i$ whenever $i' \leq i \leq k$. An element
$e_k$ is a maximal element of $[e_i,e_j]$ precisely when $pi(k) \leq i$ and $j \leq \phi(k)$, yielding the following lemma illustrated in Fig.~\ref{fig:rectangle-stabbing}.

\begin{lemma}\label{lem:stab}
An element $e_k$ is a maximal element of the set $[e_i,e_j]$ if and only if $\pi(k) \leq i \leq k \leq j \leq \phi(k)$; equivalently, $e_k$ is a maximal element of the set $[e_i,e_j]$ if and only if the point $(i,j)$ stabs the rectangle $[\pi(k), k] \times[k, \phi(k)]$.
\end{lemma}

In general, the method for computing $\phi$ is to initialize a data structure
$C$. Then process the temporal-points in order. For each temporal-point $e_k$ we query $C$ to find and set $\phi[e] = k$ for all elements $e$ less than $e_k$. Then these points are removed from $C$ and $e_k$ is added to $C$.  The computation of $\pi$ uses the same algorithm, but processes points in reverse order. With $\phi$ and $\pi$ computed, the problem is now reduced to the rectangle stabbing problem, for which we use existing data structures.

\begin{lemma}[Eppstein \emph{et al.} \cite{EppMut-SODA-01} and Agarwal \emph{et al.} \cite{AgaGovSat-ESA-2002}]
A set of $n$ rectangles whose endpoints lie on the grid $[0,n]\times [0,n]$ can be preprocessed into a data structure of size $O(n^{1+\epsilon})$ in $O(n^{1+\epsilon})$ time that can report the rectangles stabbed by a query point in $O(k)$ time and count them in $O(1)$ time, where $k$ is the number of rectangles reported. If the rectangles are colored, the set of distinct colors stabbed can be reported in time $O(k)$, where $k$ is the number of colors reported.
\end{lemma}

\begin{lemma}\label{lem:general-order}
A sequence of elements, $e_i$ for $0 \leq i < n$, from a preordered set can be preprocessed into a data structure of size $O(n^{1+\epsilon})$ in time $O(n^{1+\epsilon})$ that can report the maximal elements in a window $[e_i, e_j]$ in $O(k)$ time where $k$ is the number of reported elements, assuming that $\pi$ and $\phi$ can be computed in $O(n^{1+\epsilon})$ time.
\end{lemma}

The \emph{skyline} of a set of points in $\mathbf{R}^d$ is defined to be the maximal elements in the set under the dominance relation where a point $p$ is said to be dominated by a point $p'$ if $p[i] \leq p'[i]$ for $0 \leq i < d$ and $p \neq p'$. So our general method applies, for computing the skyline. Mortensen presents a dynamic data structure for dominance queries in $\mathbf{R}^d$ that supports insertion and deletion of points in $O(\log^{d} n)$ time and reporting of all points dominated by a given query point in $O(\log^{d} n + k)$ time where $k$ is the number of reported points \cite{Mor-SJC-06}. We can use this data structure to compute $\phi$ and $\pi$ for the dominance relation in $O(n\log^d n)$. So by Lemma~\ref{lem:general-order} we have following theorem.

\begin{theorem}
A sequence of temporal points, $p_i$ for $0 \leq i < n$, in $\mathbf{R}^d$ can
be preprocessed into a data structure of size $O(n^{1+\epsilon})$ in
$O(n^{1+\epsilon})$ time such that a query for the skyline of $[p_i, p_j]$ can
be reported in $O(k)$ time. Furthermore, if the points are colored then the
distinct colors on the skyline can be reported in $O(k)$ time.
\end{theorem}

\fi

\section{Convex hull in $\R^2$}
\label{sec:convex}
\ifFull
It is well-known that identifying the vertices of a convex hull of a set of $n$ points
has a lower bound of $\Omega(n\log n)$ in the algebraic computation tree model~\cite{Ben-Or:1983} in the worst case. 
 As with the skyline problem, 
However, 
the set of points on the convex hull will often be smaller than the total set of points. For this reason, an output sensitive algorithm was designed by Chan~\cite{c-osrch-96}, running in optimal $O(n \log h)$ time.
\else
%
\fi{}
\ifFull
Since there are $\Theta(n^2)$ unique windows, precomputing and storing the
convex hull for every window is prohibitively expensive in both time and
space.  Instead we preprocess a family of \emph{canonical subsets} of
events to achieve a more balanced space/time tradeoff of log-linear space and
polylogarithmic query time. 

On the surface, our strategy is related to the \emph{decomposition scheme} used to solve many decomposable geometry searching problems~\cite{ae-grs-97}. However, we build our decomposition over \emph{time} rather than the underlying geometry, and use sophisticated query algorithms to answer even those queries that are not traditionally considered decomposable, such as line stabbing.
\else
Like many computational geometry data structures, ours are based on a 
\emph{decomposition scheme}. We preprocess a family of \emph{canonical subsets} of
events to achieve a balanced space/time tradeoff of log-linear space and
polylogarithmic query time.  Although our choice of canonical subsets is
\emph{independent of the geometry} of the points, this approach yields
surprisingly natural query algorithms.  
We begin by presenting algorithms for convex hull queries that are still
decomposable, even though the decomposition is over time, including gift-wrapping
and linear programming queries. Then, through a novel combination of
sophisticated techniques, we adapt our approach to support line stabbing
queries.
\fi{}

In $\mathbf{R}^2$ computing the convex hull can be done by computing the upper hull and the symmetric problem of computing lower hull. So, when convenient we will only consider the computation of the upper hull.

\subsection{Hierarchical Decomposition}
\ifFull In our decomposition, we 
\else
We
\fi
build a balanced binary search tree $T$ over time with
a unique leaf for each temporal point (see Fig.~\ref{fig:decomp}). To each node $v \in T$ we associate a canonical subset $C_v$. If $v$ is a leaf, then $C_v = \{e_t\}$ where $e_t$ is the temporal point corresponding to $v$, otherwise $C_v$ is the union of its children's canonical subsets. We say that a node $v \in T$ \emph{covers} the temporal-point $e$ if $e \in C_v$.

We will assume that $T$ has been augmented with \emph{level-links}, pointers
between consecutive nodes at the same depth, and an array $A$ of the leaves
providing a mapping between temporal points and leaves. So, given any time
window $[p_i,p_j]$, we can find a set of $O(\log w)$ canonical sets which cover
$[p_i,p_j]$ in $O(\log w)$ time by working up in $T$ from the leaves at $A[i]$
and $A[j]$. Furthermore, at each node $v$, we store the convex hull of $C_v$ in
clockwise order beginning with the maximum point lexicographically, and we store the index of the point with the minimal coordinates lexicographically, providing access to the entire, upper, and lower hulls.

\ifFull
\begin{figure}[t]
\centering
\vspace{-1em}
\includegraphics[height=2in]{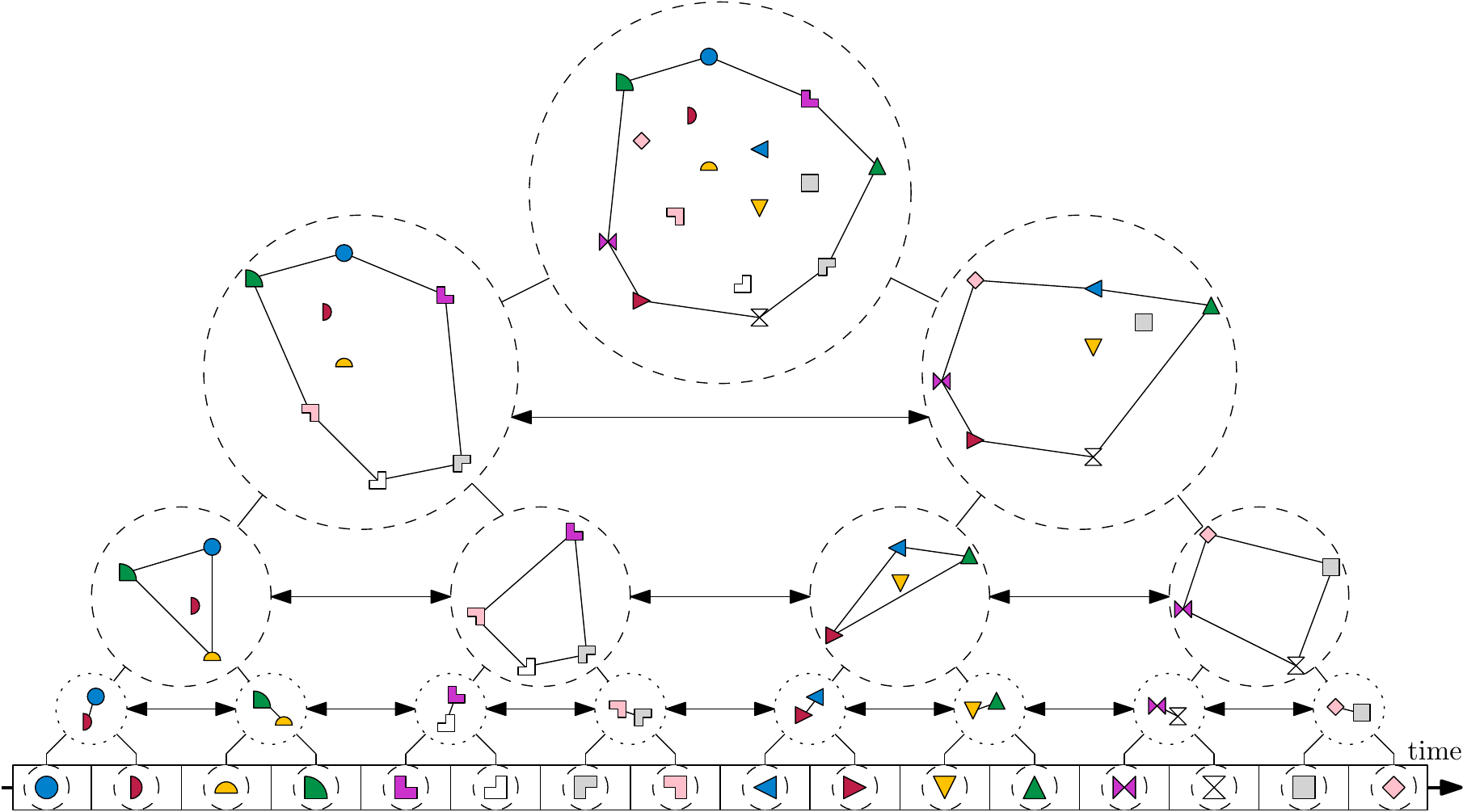}
\vspace{-1em}
\caption{
    \label{fig:decomp}
    Decomposition tree of temporal points. Each temporal point has a corresponding geometric point with coordinates in $\mathbf{R}^2$. Each node stores the convex hull of the points in its subtree.
}
\vspace{-1em}
\end{figure}
\fi

It is well known that we can find the convex
hull of a set of sorted points in linear time using the Graham
scan algorithm~\cite{graham-scan}. Therefore, we construct the tree bottom up from the
leaves. Each leaf contains a single point and no processing is required.
For each internal node $v$, we merge the sorted lists from the left and right child into
a single list for $C_v$ in $O(|C_v|)$ time. 
Then we construct the convex hull of $C_v$ also in linear time using Graham scan. 
Each point is stored in at most $O(\log n)$ nodes, and therefore the total space
required for $T$ is $O(n \log n)$. We construct each internal node in time
linear in the the number of leaves in its subtree, 
and thus the total time required to construct $T$ is $O(n \log n)$.
We summarize this result in the following lemma.

\begin{lemma}
\label{lem:decomp}
We can build a decomposition tree $T$ over $n$ events in $O(n \log n)$ time
using $O(n \log n)$ space such
that given a time window $[p_i,p_j]$, we can find $O(\log w)$ nodes of $T$ which cover $[p_i, p_j]$ in $O(\log w)$ time.
\end{lemma}

Given a window $[p_i, p_j]$ we call the $O(\log w)$ sub-hulls covering $[p_i, p_j]$ 
as \emph{canonical sub-hulls}, and we call the set of canonical sub-hulls the
\emph{canonical cover} of the window.
For all the queries in this section,
we assume the canonical decomposition has been precomputed.

\subsection{Gift Wrapping}
A classic algorithm for computing the convex hull is gift wrapping, also known
as Jarvis's March~\cite{j-ichfs-73}.  This technique starts at a point $p_i$ on
the convex hull and through a comparison of the polar angles of the other points
with respect to $p_i$ as the center, selects the point $p_{i+1}$, such that all
other points are to the right of $p_{i+1}$.  If this search is done linearly,
the next point on the hull can be found in $O(n)$ time, and thus the entire hull
can be computed in $O(nh)$ time.  
Given a point $q$ on the complete convex hull of a window, clockwise or counterclockwise \emph{Gift Wrapping} queries, locating the clockwise or counterclockwise adjacent point on the hull, can be done more quickly using our hull decomposition.

\begin{theorem}
\label{thm:gift-wrap}
Time-windowed
gift wrapping queries on the convex hull of $[p_i, p_j]$ can be answered in $O(\log^2 w)$ time.
\end{theorem}

\ifFull
\begin{figure}[t]
\vspace{-1em}
\begin{minipage}[t]{0.475\linewidth}
\centering
\includegraphics[height=1.25in]{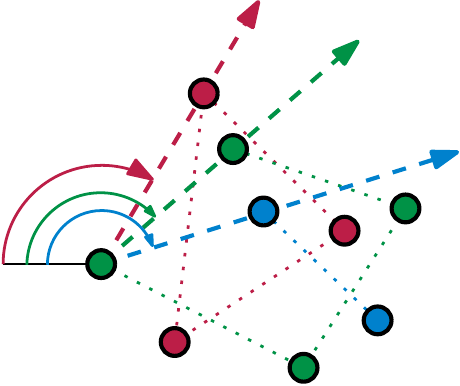}
\vspace{-1em}
\caption{
    \label{fig:giftwrap}
    A decomposed gift wrapping query.
}
\end{minipage}
\hfill
\begin{minipage}[t]{0.475\linewidth}
\centering
\includegraphics[height=1.25in]{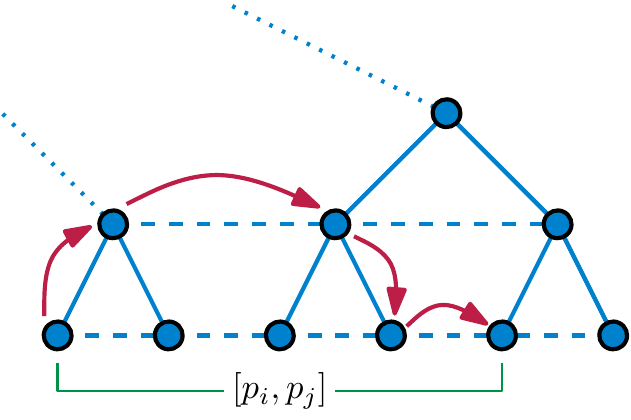}
\vspace{-1em}
\caption{
    \label{fig:window-walking}
    Walking through the window
}
\end{minipage}
\vspace{-1em}
\end{figure}
\fi

\begin{cor}\label{cor:hull-wrap}
The convex hull of $[p_i, p_j]$ can be computed in $O(h\log^2 w)$.
\end{cor}

This technique can be used to answer tangent queries as well, where a \emph{tangent query} reports the two tangents of the hull passing through a query point $q$ or an exception if $q$ is in the hull.

\begin{cor}\label{cor:tangent}
Tangent queries on the convex hull of $[p_i, p_j]$ can be answered in  $O(\log^2 w)$ time.
\end{cor}

We answer a tangent query via an iterative search; we perform
a binary search for the tangent in each sub-hull. Thus, at first it seems that
our query time can easily be sped up by a logarithmic factor via standard
fractional cascading techniques. However, the answer to a tangent query in one
sub-hull may not give us enough information about the answer to a tangent query in
another sub-hull. Recall that a convex hull partitions the plane into the region
inside the hull, and a set of wedges outside the hull,
where each wedge corresponds to the set of query points
which will all return the same convex hull point as the answer to a tangent
query 
(see Fig.~\ref{fig:hull-wedge-grid}). 
Note that by moving the query point, we can maintain the same answer to a
tangent query on one sub-hull while dramatically changing the answer to the
query on other sub-hulls. Thus, there is no clear strategy on how to preprocess
the hulls in order to leverage fractional cascading and speed up iterative
tangent queries for arbitrary query points.

\begin{figure*}
\begin{minipage}[t]{0.3\linewidth}
\centering
\vspace{-1em}
\includegraphics[height=0.9in]{gift-wrapping}
\vspace{-1em}
\caption{Gift wrapping.}
\label{fig:giftwrap}
\end{minipage}
\hfill
\begin{minipage}[t]{0.35\linewidth}
\centering
\vspace{-1em}
\includegraphics[height=0.9in]{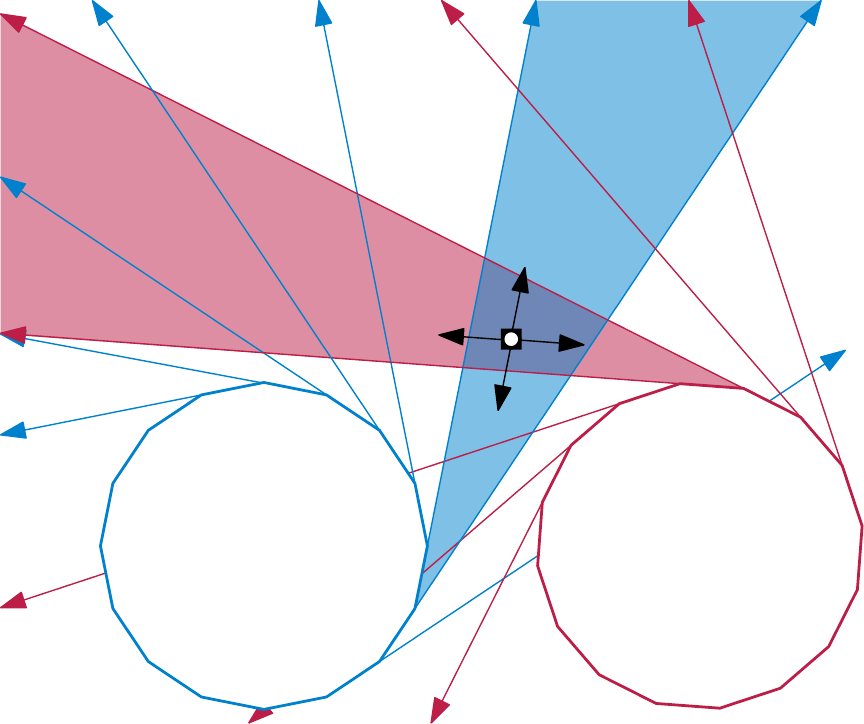}
\vspace{-1em}
\caption{Tangent query example.}
\label{fig:hull-wedge-grid}
\end{minipage}
\hfill
\begin{minipage}[t]{0.3\linewidth}
\centering
\vspace{-1em}
\includegraphics[height=0.9in]{window-walking}
\vspace{-1em}
\caption{Fractional cascading.}
\label{fig:window-walking}
\end{minipage}
\end{figure*}

\subsection{Linear Programming}
In a \emph{Linear Programming} query, we are given a direction and would like to find a point on the complete convex hull furthest in the queried direction.
\begin{theorem}\label{thm:linear-prog}
Time-windowed
linear programming queries on the convex hull of $[p_i, p_j]$ can be answered in  $O(\log w)$ time.
\end{theorem}
Additionally, we can answer the \emph{Line Decision} problem, determining if a line intersects the hull.
\begin{cor}\label{cor:decision}
Line decision queries on the convex hull of $[p_i, p_j]$ can be answered in  $O(\log w)$ time.
\end{cor}

\subsection{Line Stabbing}
For the \emph{Line Stabbing} query, a query line, $Q$, is given and we seek the
edges of the completed convex hull, if any, that intersect the line. Without
loss of generality we will consider the problem of directed line stabbing, i.e.,
we impose a direction on the query line and return the intersected edge furthest
in that direction. We observe that computing convex hull of all pairs of canonical sub-hulls is too inefficient;  since they are not separated in space, they may require a linear number of bridge 
\ifFull facets (see Fig.~\ref{fig:bad-hull-merge}).
\else
facets.
\fi

\begin{figure*}[t]
\begin{minipage}[t]{0.3\linewidth}
\centering
\vspace{-1em}
\includegraphics[height=1.35in]{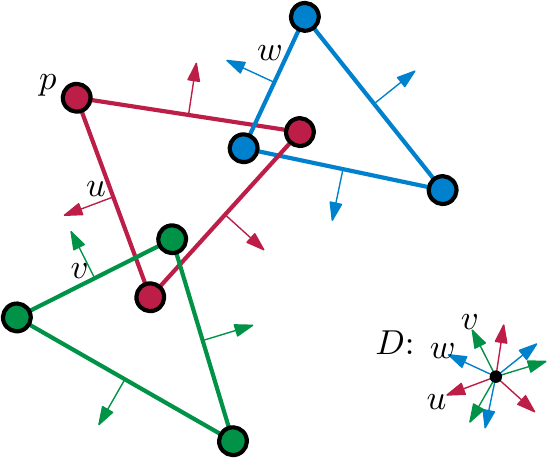}
\vspace{-1em}
\caption{
\label{fig:bet-adj}
The point $p$ is extremal for $w$ which is between $u$ and $v$.
}
\end{minipage}
\hfill
\begin{minipage}[t]{0.3\linewidth}
\centering
\vspace{-1em}
\includegraphics[height=1.35in]{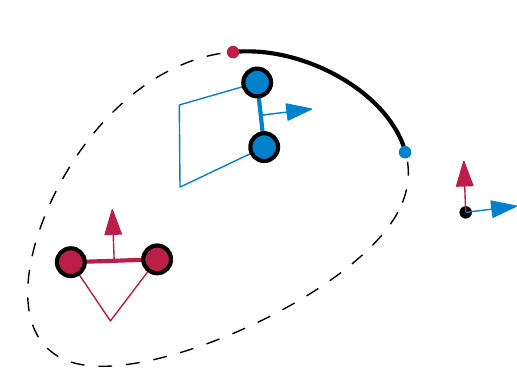}
\vspace{-1em}
\caption{The solid region of the complete hull contains the points between the two vector.}
\label{fig:bet-ext}
\end{minipage}
\hfill
\begin{minipage}[t]{0.3\linewidth}
\centering
\vspace{-1em}
\includegraphics[height=1.25in]{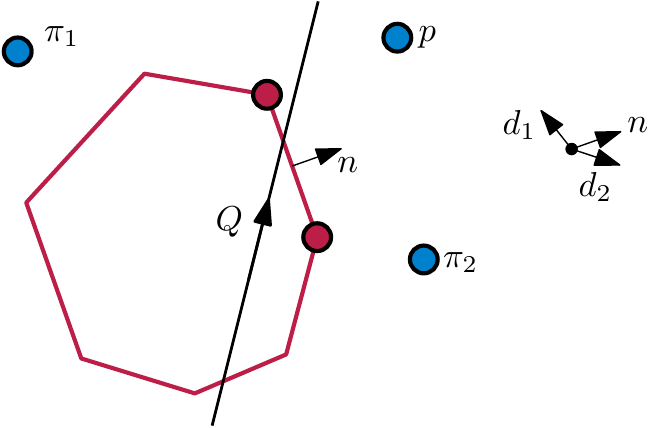}
\vspace{-1em}
\caption{The facet normal $n$ is queried and found to improve the right side bound.}
\label{fig:line-stabbing}
\end{minipage}
\end{figure*}

First, we need to define some additional terminology.  In
Figure~\ref{fig:bet-adj}, we have a few convex hulls and their facet normal
vectors.  The circular list $D$ consists of all of the normals in the sub-hulls
sorted in clockwise order.  The vector $w$ is \emph{between} $u$ and $v$ because
$w$ falls between them in the sorted order.  The point $p$ is \emph{between} the
two vectors $u$ and $v$ on a convex hull because it is extremal for a vector
between $u$ and $v$, namely $w$ or $v$ itself (see Fig.~\ref{fig:bet-ext}).
Finally, $v$ and $w$ are \emph{adjacent} vectors in $D$ if they are consecutive
in the angular order.

\begin{lemma}\label{lem:find-adjacent}
    If $Q$ intersects the complete hull, then in $O(\log^2 w)$ time we can find adjacent normals $n_1$ and $n_2$ in $D$ such that the intersected facet is between $n_1$ and $n_2$.
\end{lemma}
\begin{proof}
Let $d_1$ and $d_2$ be the left and right perpendicular directions of $Q$, respectively. Then set $\pi_i$ to be the extremal point on the complete hull in the direction $d_i$ for $i=1,2$. Then we iterate through the canonical cover, considering each canonical sub-hull. Within each sub-hull we iterate through its facet normal vectors $n$. If $n$ is between $d_1$ and $d_2$, we compute, $p$, the extremal point in the direction $n$. If $p$ is to the left of $Q$, then we set $\pi_1 = p$ and $d_1 = n$, otherwise we set $\pi_2 = p$ and $d_2 = n$
(see Fig.~\ref{fig:line-stabbing}). 
After processing all of the facet normals in all of the canonical sub-hulls $d_1$ and $d_2$ will be adjacent vectors.  
Since we have maintained the invariant that the facet crossing $Q$ is between them, $d_1$ and $d_2$ are the desired normals.
The running time of this algorithm is dominated by the $w$ extremal point queries, each taking time $O(\log w)$ time.  However we can speed up the number of linear programming queries by using weighted median selection driven prune and search. To start the search for each convex hull compute the two normals closest to the perpendiculars of $Q$. The number of vectors between these two normals will be the weight for each list and the two normals will dictate the left and right ends of the lists.  Then by choosing the weighted median of the medians, a single linear programming query can eliminate a quarter of the remaining weight.  Calculating the weighted median takes $O(\log w)$ time and querying the median takes $O(\log w)$ time.  Because the hull starts with less than or equal to $w$ total weight, it takes $O(\log w)$ queries to find the adjacent vectors.  This gives a total runtime of $O(\log^2 w)$.
\end{proof}

\ifFull
\begin{figure}[t]
\vspace{-1em}
\begin{minipage}[t]{0.475\linewidth}
\centering
\includegraphics[height=1.25in]{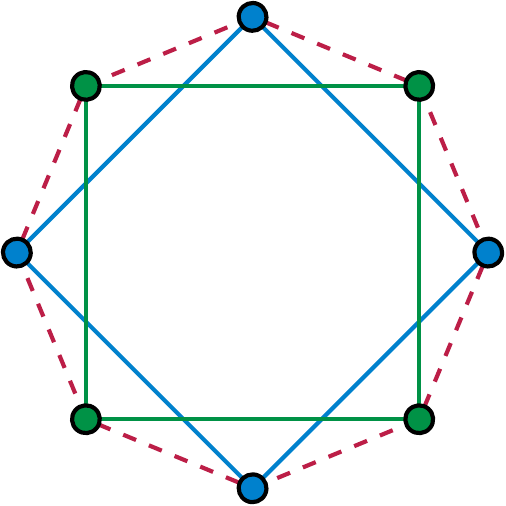}
\vspace{-1em}
\caption{Merging two canonical sub-hull may require the addition of a linear number of facets.
}\label{fig:bad-hull-merge}
\end{minipage}
\hfill
\begin{minipage}[t]{0.475\linewidth}
\centering
\includegraphics[height=1.25in]{line-stabbing}
\vspace{-1em}
\caption{The normal $n$ of an edge on a sub-hull is queried and found to improve the right side bound.}
\label{fig:line-stabbing}
\end{minipage}
\vspace{-1em}
\end{figure}
\fi

\begin{lemma}\label{lem:count-between}
For two edge normals, $u$ and $v$, that are adjacent in $D$, there are at most $3 \log w$ points between $u$ and $v$ on the sub-hulls.
\end{lemma}

\begin{theorem}
Time-windowed
line stabbing queries on the convex hull of $[p_i, p_j]$ can be answered in $O(\log^2 w)$ time.
\end{theorem}
\begin{proof}
To first establish if the line hits the convex hull we will run two extremal point queries in the directions perpendicular to the line, similar to how we solved the line decision problem.
Then we use Lemma~\ref{lem:find-adjacent} to find adjacent vectors in $D$ surrounding the edge we seek, in $O(\log^2 w)$ time. Now, Lemma~\ref{lem:count-between} implies there are $O(\log w)$ points between these adjacent vectors on the canonical sub-hulls. So we have $O(\log^2 w)$ pairs to check.  Thus the above algorithm answers line stabbing queries in $O(\log^2 w)$ time. 
\end{proof}

\emph{Vertical Line Stabbing} queries are the special case of line stabbing where the query lines are oriented vertically.  \emph{Membership} queries are to given a point, $p$, decide if $p$ is on the edge of the completed convex hull.  These can be contrasted with \emph{Containment} queries which ask whether a point $p$ is contained by the completed convex hull. 

\begin{cor}\label{cor:containment}
Vertical line stabbing, membership, and containment queries on the convex hull of $[p_i, p_j]$ can also be answered in $O(\log^2 w)$ time.
\end{cor}

\section{Proximity queries}\label{sec:proximity}
In this final section we will consider windowed queries based on their proximity. This includes approximate nearest neighbor queries and the construction of proximity graphs.

\subsection{Preliminaries}
The $Z$-order (or Morton order) is a linear ordering of the points in $\mathbf{R}^2$ introduced by Morton in 1962~\cite{Mor-TR-66}. This ordering can be described in many ways, but for our purposes it is best understood as the depth-first traversal order of points in a quadtree. We will denote this linear ordering by the symbol $\leq_Z$. Considering points in their $Z$-order is a dimension reduction technique that is often used for proximity based data 
\ifFull
structures~
\cite{GooSi-ISAAC-2011, BerEppTen-IJCGA-1999, Sam-ACMCA-1984, Cha-PC-97, DerSheSle-CCCG-2008, LiaLopLeu-DE-2001}.
\else
structures~
\cite{GooSi-ISAAC-2011, BerEppTen-IJCGA-1999}.
\fi

\begin{lemma} \label{lem:z-cell}
Let $P$ be set of points stored in a quadtree, and $C$ a specific quadtree cell storing the points $z_0, \ldots, z_k$ in $Z$-order. If $p$ is a point in $P$ with $z_0 \leq_Z p \leq_Z z_k$, then $p$ is in $C$
\end{lemma}

\subsection{Hierarchical decomposition}
To support proximity queries,
we will build a decomposition tree over time, as we did for convex hulls.  
As before, each node $v$ in the tree corresponds to a canonical subset $C_v$,
consisting of the points associated with the leaves in its subtree. However, for
proximity queries we will be storing the $Z$-order for each of the canonical
subsets\ifFull (see Fig~\ref{fig:ann-decomp}).\else.\fi{} This is equivalent to using $2$-dimensional range tree where the first coordinate is a point's time and the second coordinate is its position in $Z$-order. We will use the standard fractional cascading techniques to speed up queries~\cite{Chazelle-Guibas:1986}. In addition to the $Z$-order, we augment each internal node $v$ with a skip-quadtree $Q_u$ built over the points in $C_v$. For each cell of the quadtree we store the first and last points in the cell according to their $Z$-order.
The proof of the following lemma is given in the appendix.
\ifFull
\begin{figure}[t]
\centering
\includegraphics[height=3in]{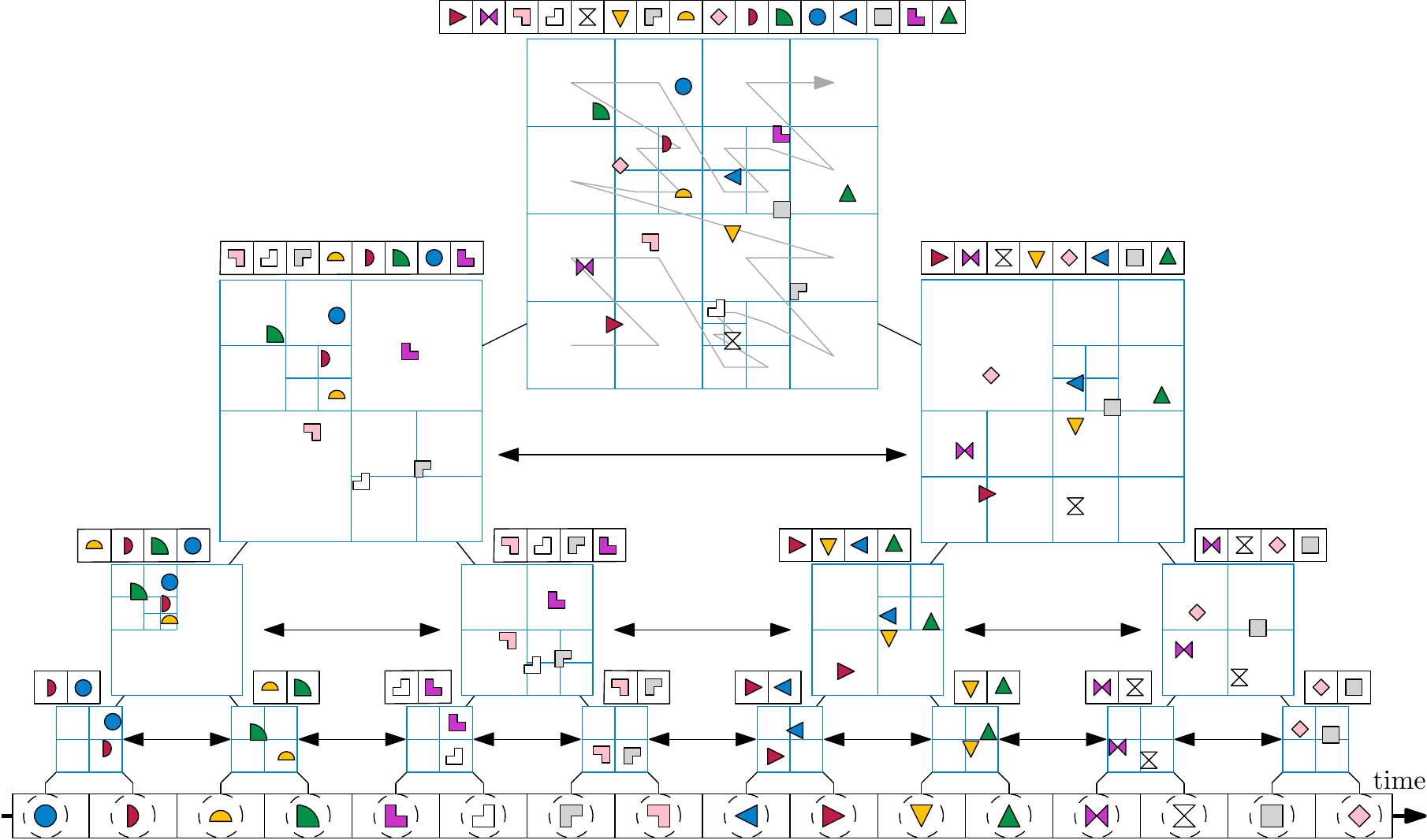}
\caption{\label{fig:ann-decomp}Decomposition tree over temporal points. Each temporal point has a corresponding geometric point with coordinates in $\mathbf{R}^2$. Each node stores the quadtree and z-order of the points in its subtree.}
\end{figure}
\fi
\begin{lemma}\label{lem:range-finding}
Any query window of width $w$ can be covered by two canonical subsets $C_1$ and $C_2$ each of width less than $2w$. Moreover, we can find $C_1$ and $C_2$ in $O(\log w)$ time. 
\end{lemma}

\subsection{Approximate spherical range searching queries}
In an \emph{approximate spherical range searching query} a query point $q$ and radius $r$ are given, and all points whose distance to $q$ is less than or equal to $r$ must be returned and no points whose distance to $q$ is greater than $(1+\epsilon)r$ may be returned, where $\epsilon$ is a fixed constant. The points whose distance to $q$ is between $r$ and $(1+\epsilon)r$ may or may not be reported. 
Due to space constraints, the proof of the following theorem is given in the
appendix.
\begin{theorem} \label{lem:arrq}
    Approximate time-windowed 
$d$-dimensional spherical range reporting queries can be performed in $O(\log w +k)$ time, for any fixed dimension $d\ge 2$.
\end{theorem}
In our definition of an approximate range query we are assuming regions
are perfectly spherical. 
However, our results can be extended to more general regions using known techniques~\cite{EppGooSun-SOCG-05}.
Note that in the special case where the query range is an axis aligned rectangle,
we can answer an exact orthogonal range query in optimal time using known techniques.
In the \emph{orthogonal range searching problem} we are given a collection points in the plane and an axis aligned query rectangle from which we must report the set of points contained within the rectangle. 
Alstrup~\emph{et al.}~\cite{abr-ndsfors-00} give a solution for orthogonal range
searching in $\R^3$ using $O(n \log^{1 +\epsilon} n)$ space and $O(\log n + k)$
query time. Simply by treating time as a spatial dimension, this allows us to
answer windowed $2$-dimensional orthogonal range searching queries. 

\subsection{Approximate nearest neighbor queries}
Given a set of points $P$ and query point $q$ an $\epsilon$-approximate nearest neighbor query asks for a point $p$ in $P$ whose distance to $q$ is at most $(1+\epsilon)r$ where $r$ is the distance to $q$'s nearest neighbor in $P$. 
The following lemma establishes a relationship between approximate nearest
neighbor queries on multidimensional points and \emph{successor} queries in the
$Z$-order of those points.
Recall that in the \emph{successor query problem} we are given a set of points $A$ on the real line and a query point $q$ on the line, and we must report the smallest element in $A$ greater than $q$. 

\begin{lemma}[Liao, Lopez and Leutenegger~\cite{LiaLopLeu-DE-2001}]
\label{lem:cANN}
Let $P$ be a set of points in $R^d$. Define a constant $c = \sqrt{d}(4d+4) + 1$. 
Suppose that we have $d+1$ shifted lists $P + v^{j}$ for $j= 0, \ldots, d$ (the specific values of $v^j$ are given in~\cite{LiaLopLeu-DE-2001}) , each one sorted according to its $Z$-order.  We can find a query point $q$'s $c$-approximate nearest neighbor in $P$ by examining  the $2(d+1)$ predecessors and successors of $q$ in the lists.
\end{lemma}


In the windowed model, a successor query corresponds to a two-dimensional
geometric query, where the time of a point maps to its $x$ coordinate,
and the value of the point maps to its $y$ coordinate. To find the
successor of value $q$ in window $[t_1, t_2]$, we slide a horizontal line
segment $[(t_1, q), (t_2, q)]$ upward, and the first point we hit is the answer
(see Fig.~\ref{fig:window-successor}).

We can answer windowed successor queries in $O(n \log n)$ space and $O(\log w)$ time per
query using a structure similar to a $2$-d range tree. 
We build a decomposition tree over the time of the points, where each internal node
stores the points in its canonical set sorted by value. 
We answer a query by performing a successor
query at each of the $O(\log w)$ nodes which together cover the window, and we
speed up the iterative queries at the internal nodes using fractional cascading.

Now, we can leverage our windowed successor data structure to answer windowed
approximate nearest neighbor queries. 
In each node of our decomposition tree we store $d+1$ copies of its canonical set, 
sorted according to the $d+1$ shifted versions of the $Z$-order from Lemma~\ref{lem:cANN}.
Given a query point $q$ and window $W = [p_i, p_j]$, we find the windowed
successor and predecessor in each of the shifted $Z$-orders. By
Lemma~\ref{lem:cANN}, one of the $2(d+1)$  points which is closest points is
guaranteed to be a $c$-approximate nearest neighbor of $q$. 

Now, we refine our answer to an $\epsilon$-nearest neighbor 
via a binary search over potential distances to the $\epsilon$-nearest neighbor. 
 The refinement process requires $O(\log \frac{1}{\epsilon}) = O(1)$ approximate
 spherical emptiness queries~\cite{AryDafMou-ESA}.
 For this strategy to work in the windowed model, we must support
 \emph{windowed} approximate spherical emptiness queries, 
 which can be done with minor modifications to our
windowed approximate range query structure. Namely, we report only the first
point found, or empty if the range is empty. This means the entire binary search
takes $O(\log w)$ time to complete. Thus, we have proven the following theorem.

\begin{theorem}
Approximate time-windowed nearest neighbor queries in $\R^d$ can be performed in $O(\log w)$ time,
for fixed $d\ge 2$.
\end{theorem}

\subsection{Proximity graph constructions}
Finally, we can use these methods to construct most interesting proximity graphs
in linear time.  As first step we use our $Z$-order decomposition tree to find a
set of canonical subsets exactly covering our query window, taking $O(\log w)$
time. Then we merge their $Z$-orders into a $Z$-order for the points in the
query window, taking $O(w)$ time. From the $Z$-order we compute the compressed
quadtree from of the points in $O(w)$ time~\cite{Cha-IPL-2008}, and from the
compressed quadtree we compute the well-separated pair decomposition also in
linear $O(w)$ time~\cite{lm-tsst-12}. With the well-separated pair decomposition the
following can be constructed in $O(w)$ time: Delaunay triangulation, minimum spanning tree,
nearest neighbor graph and Gabriel graph~\cite{lm-tsst-12}.

\ifFull
\section{Conclusion}
We have provided a diverse algorithmic study of methods for querying
time-sequenced geometric data, and within this model we have computed skylines, convex hull, and proximity queries. Our preprocessing time and space is far superior to naive calculation and storage of the solutions for all possible windows, and the algorithms for satisfying queries on a given window are significantly faster than building the results simply from the point information.

Left as open problems are reductions in the complexity of vertical stabbing and membership on the convex hull.  Currently the runtimes for these queries match general convex hull line stabbing and containment, respectively.  It has been shown in some models that vertical stabbing and membership are structurally simpler than their counterparts~\cite{Demaine-Patrascu:2007}, and thus it may be possible to find speedup in the windowed event model.  Additionally, lower bounds for the queries in our model remain to be demonstrated.

\fi

\section*{Acknowledgements} 
This research was supported by 
ONR grant N00014-08-1-1015 
and NSF grants 0953071, 1011840, and 1228639.

{\small
    \raggedright
\bibliographystyle{splncs03}
\bibliography{windowed-geometry}}

\clearpage
\appendix

\section{Omissions from Section~\ref{sec:convex}}

\begin{proof}[Theorem~\ref{thm:gift-wrap}]
We first compute the canonical cover. Then for each of the canonical sub-hulls we perform a binary search over its points, locating the point $p$ such that the vector from $q$ to $p$ forms the maximal angle with the positive $x$-axis.
 Then, from these $O(\log w)$ points, we choose the point with maximal angle,
 completing our gift wrapping query (see Fig.~\ref{fig:giftwrap}).
The total time is dominated by the $O(\log^2 w)$ time for the binary searches. A counterclockwise gift wrapping query is answered in a similar manner.
\end{proof}

\begin{proof}[Corollary~\ref{cor:hull-wrap}]
To compute the full convex hull of $[p_i, p_j]$, we begin by locating a point we know to be on the hull, e.g., the point with the lowest $x$-value, which can be done in $O(\log w)$.  
We then perform gift wrapping queries in $O(\log^2 w)$ time per query until the whole hull is returned. 
 Since we perform one query per point on the hull, we perform $h$ before returning to our starting point.  Therefore we can compute the entirety of the hull in $O(h\log^2 w)$ time.
\end{proof}

\begin{proof}[Corollary~\ref{cor:tangent}]
We first perform a containment query (Corollary~\ref{cor:containment}) in $O(\log^2 w)$ time, returning an exception if the query point is contained in the hull. Now, given a point $q$ outside the convex hull we suppose that $q$ is on the hull, and performing a gift wrapping in clockwise and counterclockwise directions in $O(\log^2 w)$ time. Producing the requested tangents.
\end{proof}

\begin{proof}[Theorem~\ref{thm:linear-prog}]
For this query it suffices to solve the problem for each of the canonical
sub-hulls, and then return the solution that is furthest in the query direction.
Since this is an iterative searching problem,  we can use \emph{fractional
cascading}~\cite{Chazelle-Guibas:1986}. For our decomposition tree our catalog
graph is of $O(1)$ degree, as each decomposition node connects to at most one
ancestor, two children, and a left and right node via level-links (see
Fig.~\ref{fig:decomp}). This allows us to construct a recursive relation between
augmented catalogs, in which we share every sixteenth element and create
sufficient bridge pointers to allow constant time subsequent searches, while
still maintaining storage and preprocessing proportional to the size of the
decomposition node. This cascading structure allows us to solve the query in
$O(\log w)$ time. This query does not take $O(\log n)$ time because, beginning
at the edge of our window, we can navigate through our fractional cascading
structure without routing outside of the window space (see
Fig.~\ref{fig:window-walking}). 
\end{proof}

\begin{proof}[Corollary~\ref{cor:decision}]
 This query reduces to two extremal point queries in the directions perpendicular to the line.  The two extremal points are separated by the query line if and only if the line intersects the convex hull.
\end{proof}

\begin{proof}[Lemma~\ref{lem:count-between}]
Suppose there are four or more points in a sub-hull between the two edge normals.  Then there are at least three edges in that sub-hull with normals between $u$ and $v$ one of which must be strictly between $u$ and $v$.  However $u$ and $v$ were adjacent in the complete list of edge normals.  This is a contradiction so there must be less than four points between $u$ and $v$ in each sub-hull.  Because there are at most $\log w$ sub-hulls, the total number of points between $u$ and $v$ is $3 \log w$.
\end{proof}

\begin{proof}[Corollar~\ref{cor:containment}]
Vertical line stabbing queries are solved by the above algorithm.  Membership and containment queries can both be answered with line stabbing queries of lines passing through the query points.  If the two edges found surround $p$ then we know it is contained in the convex hull and if $p$ is on either edge then we know it is a member of the hull.
\end{proof}

\section{Omissions from Section~\ref{sec:proximity}}

\begin{proof}[Lemma~\ref{lem:range-finding}]
Let $W = [p_i, p_j]$ be a window of width $w$. Set $C_1$ to be the largest canonical subset containing $p_i$ of width less than $2w$, and set $C_2$ to be the canonical subset adjacent to $C_1$ in level link list in the direction of increasing time. These sets can be found by following parent pointers for at most $O(\log w)$ levels. Finally since the widths of $C_1$ and $C_2$ are at least $w$ they cover $W$.
\end{proof}

\begin{proof}[Theorem~\ref{lem:arrq}]
Let $W = [p_i, p_j]$ be a window of width $w$. To perform a query with $W$ we first use Lemma~\ref{lem:range-finding} to find two canonical subsets $C_u$ and $C_v$ covering our query window, corresponding to nodes $u$ and $v$ in the decomposition tree, in $O(\log w)$ time. Then we search the quadtrees $Q_u$ and $Q_v$ to find their respective sets of inner cells (cells entirely contained in the approximate region) $\mathcal{I}_u$ and $\mathcal{I}_v$ each set of size $O(\epsilon^{1-d})$ where $d$ is the dimension. This can be done in $O(\log w + \epsilon^{1-d}) = O(\log w)$ time~\cite{EppGooSun-SOCG-05}.

Then for each of the inner cells $I \in \mathcal{I}_u \cup \mathcal{I}_v$ we perform a $2$-dimensional range reporting query with the rectangle $[p_i, p_j] \times [z_0,z_1]$, where $z_0$ and $z_1$ are the first and last point in $I$ in $Z$-order,
and record the union of the points reported in $O(\log w + k)$ time where $k$ is output size. 
By a simple packing argument, the total number of inner cells is bounded by a
function of the constants $\epsilon$ and $d$~\cite{GooSi-ISAAC-2011}.
Thus, the total time for the query is $O(\log w + k)$.

By construction we have that $W \subseteq C_u \cup C_v$, i.e. all points in the
window $W$ are contained in leaves of either $T_u$ or  $T_v$. Furthermore, the set of inner cells produced contain all geometric points in the approximate query range independent of their time stamp. 
The query rectangle at each inner cell guarantees that we return precisely the
set of points contained in that cell which also have timestamps in the query window.
Thus, returned points are exactly the set of points which are both temporally in the window $W$ and geometrically in the approximate range.
\end{proof}

\section{Omitted Figures}

Some figures which are not legible in the two column format are included on the
following page.

\begin{figure}[h!b]
\centering
\includegraphics[width=.19\textwidth]{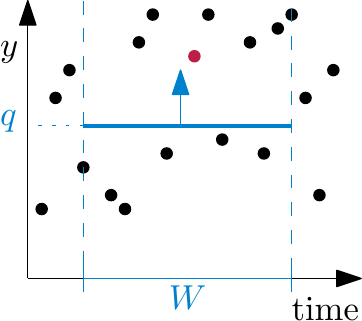}
\caption{\label{fig:window-successor}
    Windowed successor query.}
\end{figure}

\begin{figure*}[b!]
\centering
\includegraphics[width=.8\linewidth]{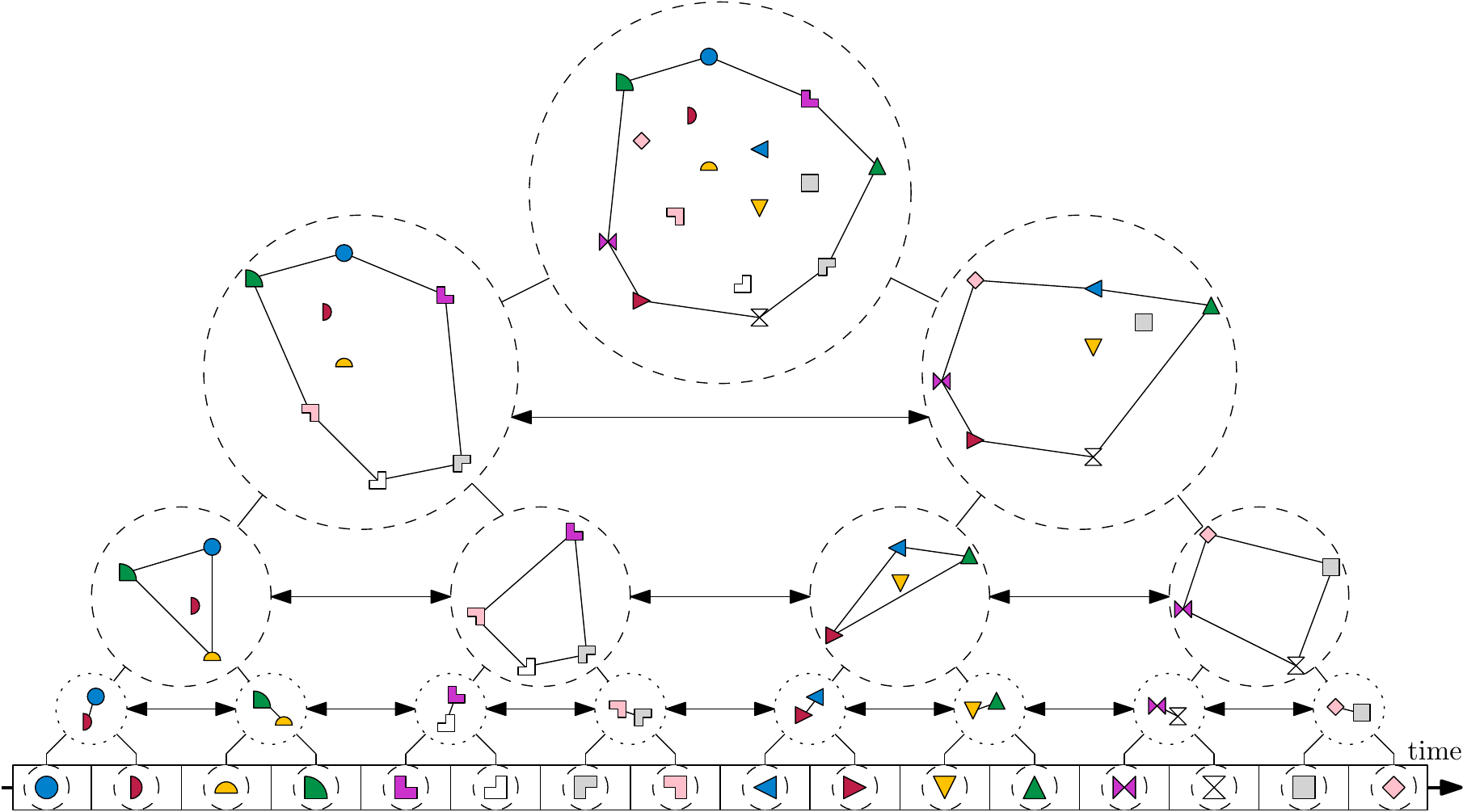}
\caption{
    \label{fig:decomp}
    Decomposition tree of temporal points. Each temporal point has a corresponding geometric point with coordinates in $\mathbf{R}^2$. Each node stores the convex hull of the points in its subtree.
}
\end{figure*}

\begin{figure*}[b!]
\centering
\includegraphics[width=.8\linewidth]{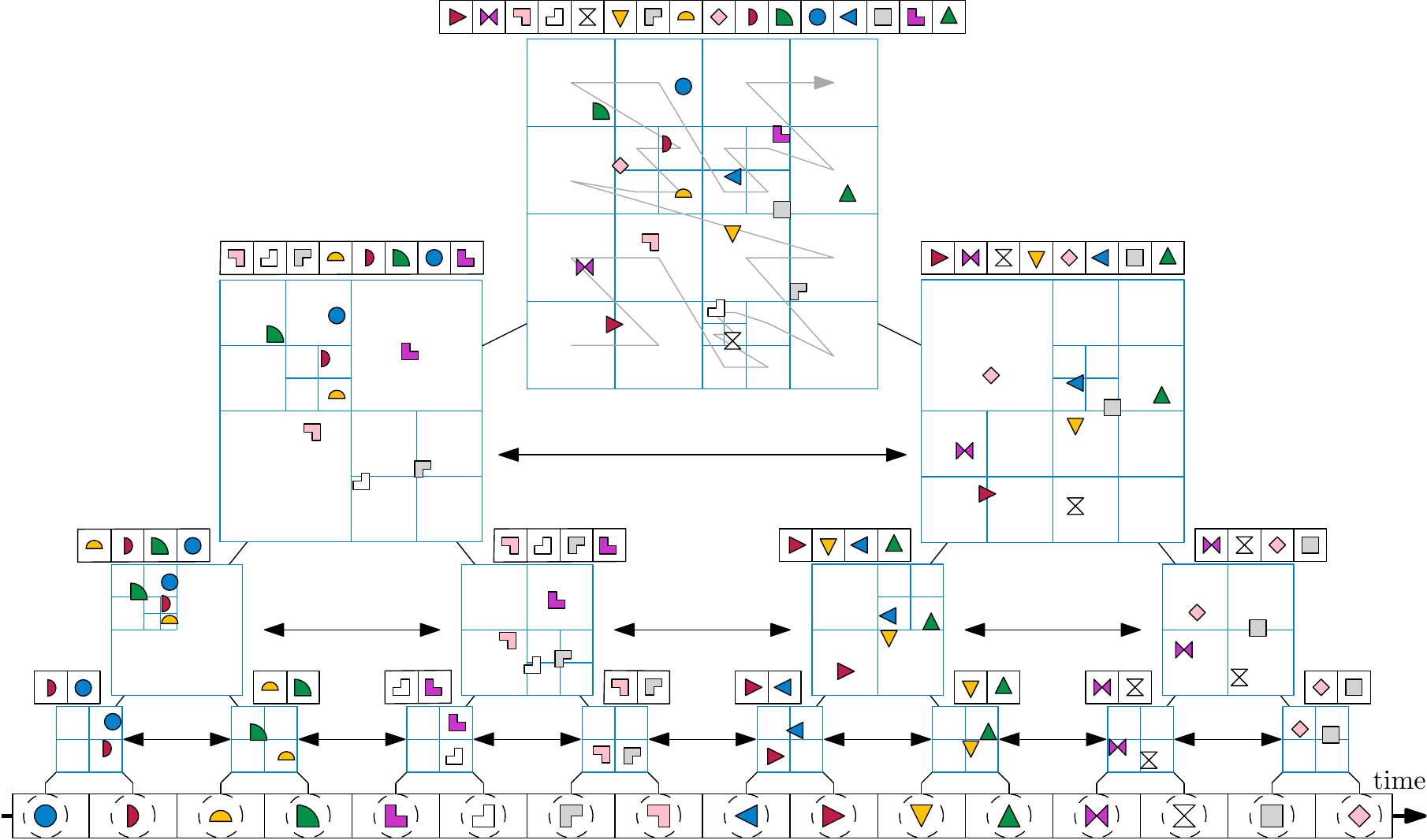}
\caption{\label{fig:ann-decomp}Decomposition tree over temporal points. Each temporal point has a corresponding geometric point with coordinates in $\mathbf{R}^2$. Each node stores the quadtree and z-order of the points in its subtree.}
\end{figure*}

\end{document}